\documentclass[11pt]{article}
\usepackage{amsmath,amssymb,amsthm}

\newcommand{\eps}{\varepsilon}
\newcommand{\med}{\mathsf{med}}

\usepackage[framemethod=TikZ]{mdframed}
\newcounter{Frame}

\newenvironment{customizedFrame}[1]
  {\mdfsetup{
    frametitle={\colorbox{white}{\space#1\space}},
    skipabove=\baselineskip plus 2pt minus 1pt,
    skipbelow=\baselineskip plus 2pt minus 1pt,
    frametitleaboveskip=-0.7\ht\strutbox,
    frametitlealignment=\center
    }
  \begin{mdframed}%
  }
  {\end{mdframed}}

\usepackage{tikz}
\usetikzlibrary{arrows}

\makeatletter
\newcommand\@erelb@r[1]{%
  \mathrel{\tikz[baseline=-.5ex]\draw[#1] (0,0)--(0.3,0);}
}
\newcommand{\erelbar}[1]{\@erelbar#1}
\def\@erelbar#1#2{%
  \ifcase\numexpr#1*4+#2\relax
    \@erelb@r{-}\or     
    \@erelb@r{->}\or    
    \@erelb@r{-|}\or    
    \@erelb@r{->|}\or   
    \@erelb@r{<-}\or    
    \@erelb@r{<->}\or   
    \@erelb@r{<-|}\or   
    \@erelb@r{<->}\or   
    \@erelb@r{|-}\or    
    \@erelb@r{|->}\or   
    \@erelb@r{|-|}\or   
    \@erelb@r{|<->|}\or 
    \@erelb@r{|<-}\or   
    \@erelb@r{|<->}\or  
    \@erelb@r{|<-|}\or  
    \@erelb@r{|<->|}    
  \else
    \@wrong
  \fi
}
\makeatother

\theoremstyle{plain}
\newtheorem{theorem}{Theorem}[section]

\newtheorem{lem}[theorem]{Lemma}

\newtheorem{clm}[theorem]{Claim}
\newtheorem{cond}[theorem]{Condition}

\newcommand{\poly}{{\rm poly}}
\newcommand{\N}{\mathbb{N}}   
\newcommand{\calP}{\mathcal{P}} \newcommand{\calD}{\mathcal{D}} 
\newcommand{\calQ}{\mathcal{Q}}

\newcommand{\OO}{\mathcal{O}}

\newcommand{\pr}[1]{\text{\bf Pr}\normalfont\bigl[ #1 \bigr]}

\newcommand{\ppr}[2]{\underset{#1}{\text{\bf Pr}} \normalfont\bigl[ #2 \bigr]}
\newcommand{\bpr}[1]{\text{\bf Pr}\normalfont\Bigl[ #1 \Bigr]}

\newcommand{\R}{\mathbb{R}}

\DeclareMathOperator*{\argmin}{arg\,min}

\newcommand{\var}[1]{\text{\bf Var}\normalfont\lbrack #1 \rbrack} 
\newcommand{\e}[1]{\text{\bf E}\normalfont\lbrack #1 \rbrack}

\newcommand{\opt}{\mathsf{opt}}
\newcommand{\Tournament}{\texttt{Tournament}}
\newcommand{\QuadTree}{\texttt{QuadTree}}
\usepackage{microtype}
\usepackage{graphicx}
\usepackage{subfigure}
\usepackage[utf8]{inputenc}
\usepackage{booktabs} 

\usepackage{hyperref}

\providecommand{\email}[1]{\href{mailto:#1}{\nolinkurl{#1}\xspace}}

\hypersetup{
     colorlinks   = true,
     citecolor    = {blue!50!teal},
     linkcolor		= purple
}
\usepackage{fullpage}

\usepackage{amsmath}
\usepackage{amssymb}
\usepackage{mathtools}
\usepackage{amsthm}

\usepackage[capitalize,noabbrev]{cleveref}

\usepackage[textsize=tiny]{todonotes}

\usepackage{subfiles}
\usepackage{url,color}
\usepackage{footnote}
\usepackage{threeparttable}
\usepackage{mdframed}
\usepackage{xspace}
\usepackage{pifont}
\usepackage{stmaryrd}
\makesavenoteenv{tabular}
\makesavenoteenv{table}
\newtheorem*{theorem*}{Theorem}

\usepackage{booktabs}
\usepackage{tablefootnote}

\usepackage{array}
\newcolumntype{L}[1]{>{\raggedright\arraybackslash}p{#1}} 
\newcolumntype{C}[1]{>{\centering\arraybackslash}p{#1}}    
\newcolumntype{R}[1]{>{\raggedleft\arraybackslash}p{#1}}   

\begin{document}

\newcommand{\tikzxmark}{%
\tikz[scale=0.23] {
    \draw[line width=0.7,line cap=round] (0,0) to [bend left=6] (1,1);
    \draw[line width=0.7,line cap=round] (0.2,0.95) to [bend right=3] (0.8,0.05);
}}
\DeclarePairedDelimiter\ceil{\lceil}{\rceil}
\DeclarePairedDelimiter\floor{\lfloor}{\rfloor}

\renewcommand{\subparagraph}[1]{\medskip\noindent\underline{\textit{#1}}}

\title{Even Faster Algorithm for the Chamfer
Distance}
\author{Ying Feng\thanks{MIT. 
E-mail: \email{yingfeng@mit.edu}}
\and
Piotr Indyk\thanks{MIT. 
E-mail: \email{indyk@mit.edu}}
}
\date{\today}
\maketitle




\begin{abstract}
For two $d$-dimensional point sets $A,B$  of size up to $n$, the Chamfer distance from $A$ to $B$ is defined as $CH(A,B)=\sum_{a \in A} \min_{b \in B} \|a-b\|$.
 The Chamfer distance is a widely used measure for quantifying dissimilarity between sets of points, used in many machine learning and computer vision applications. A recent work of Bakshi et al, NeuriPS'23, gave the first near-linear time $(1+\eps)$-approximate algorithm, with a running time of  $\OO(nd \log (n)/\eps^2)$. In this paper we improve the running time further,  to $\OO(nd(\log\log n+\log\frac{1}{\eps})/\eps^2)$. 
  When $\eps$ is a constant, this reduces 
  the gap between the upper bound and the trivial $\Omega(dn)$ lower bound significantly, from $\OO(\log n)$ to $\OO(\log\log n)$. 
\end{abstract}

\newpage

\section{Introduction}

For any two $d$-dimensional point sets $A,B$ of sizes up to $n$, the Chamfer distance from $A$ to $B$ is defined as 
\newcommand{\CH}{\mathsf{CH}}
\begin{equation*}
\CH(A,B) = \sum_{a \in A} \min_{b \in B} \|a-b\|
\end{equation*}

where $\|\cdot\|$ is the underlying norm defining the distance between the points. Chamfer distance and its variant, the Relaxed Earth Mover Distance~\cite{kusner2015word,atasu19a}, are widely used metrics for quantifying the distance between two {\em sets} of points. These measures are especially popular in fields such as machine learning (e.g.,\cite{kusner2015word,wan2019transductive}) and computer vision (e.g.,\cite{athitsos2003estimating,sudderth2004visual,fan2017point,jiang2018gal}). A closely related notion of ``the sum of maximum similarities'', where $\min_{b \in B} \|a-b\|$ is replaced by $\max_{b \in B} a \cdot b$, has been recently popularized by the ColBERT system~\cite{khattab2020colbert}. Efficient subroutines for computing Chamfer distances are provided in prominent libraries including Pytorch~\cite{pytorch3d}, PDAL~\cite{pdal} and Tensorflow~\cite{tensorflow}. In many applications (e.g., see ~\cite{kusner2015word}), Chamfer distance is favored as a faster alternative to the more computationally intensive Earth-Mover Distance or Wasserstein Distance.

Despite the popularity of Chamfer distance, efficient algorithms for computing it haven't attracted as much attention as algorithms for, say, the Earth-Mover Distance. The first improvement to the naive  $\OO(dn^2)$-time algorithm was obtained in~\cite{sudderth2004visual}, who utilized the fact that $\CH(A,B)$ can be computed by performing $|A|$ nearest neighbor queries in a data structure storing $B$. However, even when  the state of the art approximate nearest neighbor algorithms are used, this leads to an $(1+\epsilon)$-approximate estimator with only slightly sub-quadratic running time of $\OO \left(dn^{1+\frac{1}{2 (1+\epsilon)^2 -1}}\right)$ in high dimensions~\cite{andoni2015optimal}\footnote{All algorithms considered in this paper are randomized, and return $(1+\eps)$-approximate answers with a constant probability}. The first near-linear-time algorithm for any dimension was proposed only recently in ~\cite{BIJ24}, who gave a $(1+\epsilon)$-approximation algorithm with a running time of $\OO(dn \log(n)/\eps^2)$, for $\ell_1$ and $\ell_2$ norms.
Since any algorithm for approximating the distance must run in at least $\Omega(dn)$ time\footnote{The Chamfer Distance could be dominated by the distance from a single point $a \in A$ to $B$.}, the the upper and lower running time bounds differed by a factor of $\log(n)/\eps^2$.

\ \\
{\bf Our result:}
In this paper we make a substantial progress towards reducing the gap between the upper and lower bounds for this problem. In particular, we show the following theorem. Assume a Word RAM model where both the input coordinates and the memory/processor word size is $\OO(\log n)$ bits.\footnote{In the Appendix, we adopt the reduction of~\cite{BIJ24} to extend the result to coordinates of arbitrary finite precision.} Then:

\begin{theorem}
\label{t:main}
    There is an algorithm that, given two sets $A, B$ of $d$-dimensional points with coordinates in 
    $\{1 \ldots \mbox{\poly}(n) \}$ and a parameter $\eps>0$, computes a $(1+\eps)$-approximation to the Chamfer distance from $A$ to $B$ under the $\ell_1$ metric,  in time  \[ \OO(nd(\log\log n+\log\frac{1}{\eps})/\eps^2)) .\]
   The algorithm is randomized and is correct with a constant probability.
\end{theorem}

Thus, we reduce the gap between upper and lower bounds from $\OO(\log (n)/\eps^2)$ to \\ $\OO(\log\log n +\log\frac{1}{\eps})/\eps^2)$.

\subsection{Our techniques}
Our result is obtained by identifying and overcoming the bottlenecks in the previous algorithm ~\cite{BIJ24}. On a high level, that algorithm consists of two steps, described below. For the sake of exposition, in what follows we assume that the target approximation factor $1+\eps$ is some constant.

\ \\
{\bf Outline of the prior algorithm:} In the first step, for each point $a \in A$, the algorithm computes an estimate $\calD_a$ of the distance $\opt_a$ from $a$ to its nearest neighbor in $B$. The estimate is $\OO(\log n)$-approximate, meaning that we  $\opt_a \le \calD_a \le \OO(\log n) \opt_a$. This is achieved as follows. First,  the algorithm imposes $\OO(\log n)$ grids of sidelength $1,2, 4, \ldots$, and maps each point in $B$ to the corresponding cells. Then, for each $a$, it identifies the finest grid cell containing both $a$ and some point $b \in B$.
Finally, it uses the distance between $a$ and $b$ as an estimate $\calD_a$. To ensure that this process yields an $\OO(\log n)$-approximation, each grid needs to be {\em independently} shifted at random. We emphasize that this independence between the shifts of different grids is {\em crucial} to ensure the $\OO(\log n)$-approximation guarantee - the more natural approach of using ``nested grids'' does not work. The whole process takes $\OO(nd)$ time per grid, or $\OO(nd \log n)$ time overall.

In the second step, the algorithm estimates the Chamfer distance via {\em importance sampling.} Specifically, the algorithm samples $T$ points from $A$, such that the probability of sampling $a$ is proportional to the estimate $\calD_a$. For each sampled point $a$, the distance $\opt_a$ from $a$ to its nearest neighbor in $B$ is computed directly in $\OO(nd)$ time. The final estimate of the Chamfer distance is equal to the weighted average the $T$ values $\opt_a$ . It can be shown that if the number of samples $T$ is equal to the distortion $\OO(\log n)$ of the estimates $\calD_a$, this yields a constant factor approximation to the Chamfer distance from $A$ to $B$. The overall cost of the second step is $\OO(T nd) = \OO(nd \log n)$, i.e., asymptotically the same as the cost of the first step. 

\ \\
{\bf Intuitions behind the new algorithm:}  To improve the running time, we need to reduce the cost of each of the two steps. In what follows we outline the obstacles to this task and how they can be overcome. 

\ \\
{\em Step 1: } The main difficulty in reducing the cost of the first step is that, for each grid,  the point-to-cell assignment takes $\OO(nd)$ time to compute, so computing these $T$ assignments separately for each grid takes $\OO(nd T)$ time. And,  since each grid is independently translated by a different random vector,  the grids are not nested, i.e., a  (smaller) cell of side length $2^i$ might contain points from many  (larger) cells of side length $2^{i+1}$.
As a result,  is unclear how to reuse the point-to-cell assignment in one grid to speedup the assignment in another grid, while computing them separately takes $\OO(ndT)$ time. 

To overcome this difficulty, we abandon independent shifts and resort to $\OO(\log n)$ {\em nested} grids. 
Such grids can be viewed as forming a {\em quadtree} with $\OO(\log n)$ levels, where any cell $C$ at level $i+1$ (i.e., of side length $2^{i+1}$) is connected to $2^d$ cells at level $i$ contained in $C$.
(Note that the root node of the quadtree has the highest level $\OO(\log n)$).
Although using a single quadtree increases the approximation error, we show that using {\em two} independently shifted quadtrees retains the $\OO(\log n)$ approximation factor. That is, we repeat the process of finding the finest grid cell containing both $a$ and some point from $B$ twice, and return the point in $B$ that is closer to $a$. This amplifies the probability of finding a point from $B$ that is ``close'' to $a$, which translates into a better approximation factor compared to using a single quadtree.

We still need show that the point-to-cell assignments can be computed efficiently. To this end, we observe that for each point $a$, its assignment to all $\OO(\log n)$ nested grids can be encoded as $d$ words of length $\OO(\log n)$, or a $d \times \OO(\log n)$ bit matrix $M$. Each row corresponds to one of the $d$ coordinates, and the most significant bit of a row indicates the assignment to cells at the highest level (i.e. cells with the largest side length) with respect to that coordinate. In other words, the most significant bits of all coordinates are packed into the first column, etc. We observe that two points $a$ and $b$  lie in the same cell of side length $2^{i}$ if and only if their matrices agree in all but the last $i$ columns.
If we {\em transpose} $M$ and read the resulting matrix in the row-major order, then finding a point $b \in B$ in the finest grid cell containing $a$ is equivalent to finding $b$ that shares the longest common prefix with $a$. We show that this transposition can be done using $\OO( \log n \cdot \log \log n)$ simple operations on words, yielding $\OO(n \log n \cdot \log \log n) = \OO(n d \cdot \log \log n)$ time overall.

As an aside, we note that quadtree computation is a common task in many geometric algorithms~\cite{har2011geometric}. Although an $\OO(n)$ algorithm for this task was known for constant dimension $d$ ~\cite{C08}\footnote{Assuming that each coordinate can be represented using $\log n$ bits.},  to the best of our knowledge our algorithm is the first to achieve $\OO(n d \cdot \log \log n)$ time for arbitrary dimension. 

\ \\
{\em Step 2:} At this point we computed estimates $\calD_a$ such that $\opt_a \le \calD_a \le \OO(\log n) \opt_a$. Given these estimates, importance sampling still requires sampling $\Omega(\log n)$ points. Therefore, we improve the running time by {\em approximating} (up to a constant factor) the values $\opt_a$, as opposed to computing them exactly. This is achieved by computing $\OO(\log \log n)$ random projections of the input points, which ensures that that the distance between any fixed pair of points is well-approximated with probability $1-1/\mbox{poly}(\log n)$. We then employ these projections in a variant of the tournament algorithm of~\cite{K97} which computes $\OO(1)$-approximate estimates of $\opt_a$ for $\OO(\log n)$ sampled points $a$ in $\OO(nd \log \log n)$ time. Since the algorithm of \cite{K97} works for the $\ell_2$ metric as opposed to the $\ell_1$ metric, we replace Gaussian random projections with Cauchy random projections, and re-analyze the algorithm. 

This completes the overview of an $\OO(nd \log \log n)$-time algorithm for estimate the Chamfer distance up to a {\em constant} factor. To achieve a $(1+\eps)$-approximation guarantee for any $\eps>0$, we proceed as follows. 
 First, instead of sampling $\OO(\log n)$ points as before, we sample $\OO(\log(n)/\eps^2)$ points $a$. Then, we use the tournament algorithm to compute $\OO(1)$-approximations to $\opt_a$, as before. \footnote{Note that we could use the tournament algorithm to report $(1+\eps)$-approximate answers, but then the dependence of the running time on  $1/\eps$ would become {\em quartic}, as  the $1/\eps^2$ term in the sample size would be multiplied by another $1/\eps^2$ term in the bound for the number of projections needed to guarantee that the tournament algorithm returns $(1+\eps)$-approximate answers. } Then we use a technique called {\em rejection sampling} to simulate the process of sampling $\OO(1/\eps^2)$ points $a$ with probability proportional to $\Theta(\opt_a)$. For each such point, we compute $\opt_a$ exactly in $\OO(nd)$ time. Finally, we use the $\OO(1/\eps^2)$ sampled points $a$ and the exact values of $\opt_a$ in importance sampling to estimate the Chamfer distance up to a factor of $1+\eps$.

This concludes the overview of our algorithm for the Chamfer distance under the $\ell_1$ metric. We remark that \cite{BIJ24} also extends their result from the $\ell_1$ metric to the $\ell_2$ metric by first embedding points from $\ell_2$ to $\ell_1$ using random projections. This takes $\OO(nd \cdot \log n)$ time, which exceeds the runtime of our algorithm, eliminating our improvement. However, a faster embedding method would yield an improved runtime for the Chamfer distance under the $\ell_2$ metric. We leave finding a faster embedding algorithm as an open problem.


\section{Preliminaries}

In this paper, we consider the regime where the approximation factor $\eps \geq \sqrt{\frac{\log n}{n}}$. Note that otherwise, an $\OO(nd/\eps^2)$ time bound would be close to the runtime of a naive exact computation.

In the proof of Theorem~\ref{t:main}, we assume a Word RAM model where both the input coordinates and the memory/processor word size is $\OO(\log n)$ bits.
This model is particularly important in procedures $\texttt{Concatenate}$ and $\texttt{Transpose}$, where we rely on the fact that we can shift bits and perform bit-wise AND, ADD and OR operations in constant time.

\ \\
{\bf Notation:} For any integers $a \geq 1$, we use $[n]$ to denote the set of all integers from $1$ to $n$. For any two real numbers $a, b$ such that $a \leq b$, we use $[a, b]$ to denote the set of all reals from $a$ to $b$. Let $d$ be the dimension of points. 

For any $q \in \R^d$, define $\opt^P_q := \min_{p \in P}\lVert q-p\rVert_1$ for some subset $P$ of $\R^d$. We will omit the superscript $P$ when it is clear in the context. 
\section{Quadtree}

In Figure \ref{Figure:QuadTree}, we show an algorithm $\QuadTree$ that outputs crude estimations of the nearest distances simulatenously for a set of points. The estimation guarantee is the same as the $\mathtt{CrudeNN}$ algorithm in \cite{BIJ24}. While \cite{BIJ24} achieves this using a quadtree with $\log n$ {\it independent} levels, which naturally introduce a $\log n$ runtime overhead, we show that two compressed quadtrees with dependent levels suffice. Our construction of compressed quadtrees is a generalization of \cite{C08} to high dimensions.

\begin{figure}[ht!]
	\begin{customizedFrame}{\QuadTree}
		
		\begin{flushleft}
			\noindent {\bf Input:} Two size-$n$ subsets $Q := \{q_i\}_{i \in [n]}$ and $P := \{p_i\}_{i \in [n]}$ of a metric space $(\R^d, \lVert \cdot \rVert_1)$, such that $Q, P \subset [0, \alpha]^d$ for some bound $\alpha = \poly(n)$.
			
			\noindent {\bf Output:} A set of $n$ values $\{\calD_i\}_{i \in [n]}$, such that every $\calD_i \in \R$ satisfies $\calD_i \geq \opt^P_{q_i}$.

            \begin{enumerate}
				\item\label{Line:draw} Let $t = \lceil\log (\alpha)\rceil+1$. Sample two uniformly random points $z, z' \sim [0, 2^{t-1}]^d$. For any point $x \in [0, \alpha]^d$, define \[h(x) := (\lceil\vec{x}_1+\vec{z}_1\rceil, \lceil\vec{x}_2+\vec{z}_2\rceil, \cdots, \lceil\vec{x}_d+\vec{z}_d\rceil),\]
                \[h'(x) := (\lceil\vec{x}_1+\vec{z'}_1\rceil, \lceil\vec{x}_2+\vec{z'}_2\rceil, \cdots, \lceil\vec{x}_d+\vec{z'}_d\rceil),\]

                where $\vec{x}_i, \vec{z}_i, \vec{z'}_i$ are the $i$-th coordinates of $x, z, z'$, respectively.

                \item\label{Line:transpose} For each $x \in Q \cup P$:

                \begin{itemize}
                    \item Compute $h(x)$ and write each element of $h(x)$ as a $t$-bit binary string. Then $h(x)$ can be viewed as a $d$-by-$t$ binary matrix stored in the row-major order, whose $(i, j)$-th entry is the $j$-th significant bit of the $i$-th element of $h(x)$. 
                Transpose this matrix and concatenate the rows of the transpose. Denote the resulting binary string as $h(x)^\top$. 

                \item Similarly, compute $h'(x)^\top$.
                \end{itemize}

                \item Use $h(x)^\top$ as keys to sort all $x \in Q \cup P$. Also, use $h'(x)^\top$ as keys to sort all $x \in Q \cup P$.
                
                \item For each $q_i \in Q$:
                
                \begin{itemize}
                    \item Use the sort to find a $p \in P$ that maximizes the length $l$ of the longest common prefix of $h(q_i)^\top$ and $h(p)^\top$. Similarly, find a $p' \in P$ that maximizes the length $l'$ of the longest common prefix of $h'(q_i)^\top$ and $h'(p')^\top$.

                    \item If $l \geq l'$ then output $\mathcal{D}_i := \lVert q_i - p\rVert_1$; otherwise, output $\mathcal{D}_i := \lVert q_i - p'\rVert_1$.
                \end{itemize}
            \end{enumerate}

		\end{flushleft}
	\end{customizedFrame}
	\caption{The $\QuadTree$ Algorithm.}\label{Figure:QuadTree}
\end{figure}

\ \\
{\bf Correctness:} 
For any $x \in [0, \alpha]^d$ and any integer $k$ such that $0 \leq k \leq t$, let $h_k(x) := (\lceil\frac{\vec{x}_1+\vec{z}_1}{2^k}\rceil, \lceil\frac{\vec{x}_2+\vec{z}_2}{2^k}\rceil, $ $\cdots, \lceil\frac{\vec{x}_d+\vec{z}_d}{2^k}\rceil)$, where $z$ is the random point drawn on Line \ref{Line:draw} in Figure \ref{Figure:QuadTree}. 
Observe that $h_k(x)$ is related to the prefix of $h(x)^\top$.

\begin{clm}\label{Claim:prefix}
    Let $q, p \in [0, \alpha]^d$ be arbitrary. For any integer $k$ such that $0 \leq k \leq t$, $h_k(q) = h_k(p)$ if and only if $h(q)^\top$ and $h(p)^\top$ share a common prefix of length at least $d(t-k)$.
\end{clm}

\begin{proof}
    If $h(q)^\top$ and $h(p)^\top$ share a common prefix of length at least $d(t-k)$, then in hashes $h(q)$ and $h(p)$, the first $(t-k)$ bits of all $d$ coordinates are the same. $h_k(q)$ and $h_k(p)$ compute exactly these bits, thus $h_k(q) = h_k(p)$. The reverse direction holds symmetrically.
\end{proof}

Claim \ref{Claim:prefix} justifies using $h_k(\cdot)$'s as an alternative representation of the binary string $h(\cdot)^\top$. \cite{BIJ24} shows that $h_k$ has a locality-sensitive property, which will help us bound the distance between points.






\begin{clm}[Lemma A.4 of \cite{BIJ24}]\label{Claim:cite}
    For any fixed integer $k$ such that $0 \leq k \leq t$ and any two points $q, p \in [0, \alpha]^d$,
\[
\pr{h_k(q) \neq h_k(p)} \leq \frac{\lVert q- p \rVert_1}{2^k},
\]\[
\pr{h_k(q) = h_k(p)} \leq \exp{(-\frac{\lVert q- p \rVert_1}{2^k})},
\]
where the probabilities are over the random choice of $z$.
\end{clm}

We now show that if two points have the same hash $h_k$, then their distance is likely not too much greater than $2^k$. A straight-forward bound follows from the diameter of the $d$-dimensional cube.

\begin{lem}\label{Lemma:cube}
    For all $q \in Q$, $p \in P$, and $0 \leq k \leq t$, the following always holds: If $h_k(q) = h_k(p)$ then $\lVert  q- p \rVert_1  \leq 2^k\cdot d$.
\end{lem}

\begin{proof}
    Observe that $h_k(q) = h_k(p)$ only if $q+z$ and $p+z$ are in the same $d$-dimensional cube of side-length $2^k$. The diameter of such a cube under the $\ell_1$ norm is $2^k\cdot d$. Therefore, for any $q, p$ and $0 \leq k \leq t$, $\lVert q - p \rVert_1 \leq 2^k \cdot d$ is a necessary condition for $h_k(q) = h_k(p)$ to hold.
\end{proof}

Moreover, using Claim \ref{Claim:cite}, we can bound this ratio with respect to $n$.

\begin{lem}\label{Lemma:logn}
With probability at least $1-\OO(1/n)$, the following holds simultaneously for all $q \in Q$, $p \in P$, and $0 \leq k \leq t$: If $h_k(q) = h_k(p)$ then $\lVert  q- p \rVert_1  \leq 2^k\cdot 3\log n$.
\end{lem}

\begin{proof}
    We show the contrapositive that with probability $1-\OO(1/n)$, $k < \log(\lVert q- p \rVert_1 / 3\log n)$ implies $h_k(q) \neq h_k(p)$ simultaneously for all $q \in Q$ and $p \in P$. It suffices to argue that for any fixed pair of points $q \in Q$ and $p \in P$, this holds with probability at least $1-\OO(1/n^3)$. The lemma then follows by a union bound over $n^2$ pairs.

    Let $k_0$ denote the largest integer $k$ that satisfies $k < \log(\lVert q- p \rVert_1 / 3\log n)$. Then we have
    \[
    \pr{h_{k_0}(q) = h_{k_0}(p)} \leq \exp{(-\frac{\lVert q- p \rVert_1}{2^{k_0}})} \leq \exp(-3\log n).
    \]
    
    i.e., with probability at least $1-\OO(1/n^3)$, $h_{k_0}(q) \neq h_{k_0}(p)$. Also, it is easy to see that if $h_{k_0}(q) \neq h_{k_0}(p)$, then for all $k \leq k_0$, $h_{k}(q) \neq h_{k}(p)$, concluding the claim.
    
\end{proof}


Symmetrically, if we define $h'_k(x) := (\lceil\frac{\vec{x}_1+\vec{z'}_1}{2^k}\rceil, \lceil\frac{\vec{x}_2+\vec{z'}_2}{2^k}\rceil, \cdots, \lceil\frac{\vec{x}_d+\vec{z'}_d}{2^k}\rceil)$, the  claims and lemmas above also hold for $h'_k$.
Using these, we show that the expected outputs of the $\QuadTree$ algorithm are (crude) estimations of the nearest neighbor distances. 

\begin{theorem}\label{Theorem:quadtree}
    With probability at least $1 - \OO(1/n)$, it holds for all $q_i \in Q$ that $\e{\mathcal{D}_i} \leq 5 \min (d, 3\log n) \cdot \opt_{q_i}^P$.
\end{theorem}

\begin{proof}
We assume the success case of Lemma \ref{Lemma:logn} for both $h_k$ and $h'_k$. Fix an arbitrary $q_i \in Q$. Recall that the $\QuadTree$ algorithm finds $p, p' \in P$ for $q_i$, which are associated with longest common prefixes of lengths $l, l'$, respectively. For integer $k: 0 \leq k \leq t$, let $\mathcal{E}_k$ denote the event $d(t-k) \leq \max{(l, l')} < d(t-k + 1)$. Observe from Claim \ref{Claim:prefix} that when $\mathcal{E}_k$ happens,
\begin{itemize}
    \item either $l \geq l'$ and $h_k(q_i) = h_k(p)$, 
    \item or $l' > l$ and $h'_k(q_i) = h'_k(p')$.
\end{itemize}

In both cases, we know from Lemma \ref{Lemma:cube} and \ref{Lemma:logn} that $\calD_i \leq 2^k \cdot \min (d, 3\log n)$.

Let ${{D}} := \min (d, 3\log n)$, $p^* := \argmin_{p \in P} \lVert q_i-p\rVert_1$, and $k^* := \lceil\log(\opt_{q_i})\rceil$. We have
    \begin{align*}
    \e{\mathcal{D}_i} &\leq \sum_{0 \leq k \leq t} \pr{\mathcal{E}_k} \cdot (2^k \cdot {{D}}) \\
    &\leq {{D}}(\sum_{0 \leq k \leq k^*} \pr{\mathcal{E}_k } \cdot \opt_{q_i}  + \sum_{k^* < k \leq t} \pr{h_{k-1}(q_i) \neq h_{k-1}(p^*) \wedge h'_{k-1}(q_i) \neq h'_{k-1}(p^*)} \cdot 2^k)
\end{align*}
where the second inequality holds because $\mathcal{E}_k$ implies that neither pair $\{h(q_i)^\top, h(p^*)^\top\}$ nor $\{h'(q_i)^\top, h'(p^*)^\top\}$ share a common prefix of length $\geq d(t-k+1)$. Thus $h_{k-1}(q_i) \neq h_{k-1}(p^*)
$ and $h'_{k-1}(q_i) \neq h'_{k-1}(p^*)$ by Claim \ref{Claim:prefix}. 

Moreover, events $\mathcal{E}_k$ for all $k$ form a partition of a sample space, so $\sum_k \pr{\mathcal{E}_k } \leq 1$.  Applying this and the locality sensitive properties of $h_{k-1}$ and $h'_{k-1}$, we get
\[ 
    \e{\mathcal{D}_i} \leq {{D}}(\opt_{q_i} + \sum_{k^* < k \leq t} (\frac{\opt_{q_i}}{2^{k-1}})^2 \cdot 2^k) 
    \leq {{D}}(\opt_{q_i} + 2\opt_{q_i}\sum_{k^* < k \leq t} \frac{\opt_{q_i}}{2^{k-1}}) 
    \leq 5{{D}} \cdot \opt_{q_i}
\]
\end{proof}
\ \\
{\bf Runtime analysis:}  


\begin{lem}[Line \ref{Line:transpose}]
    For any $x \in Q\cup P$, $h(x)^\top$ (and $h'(x)^\top$) can be computed in $\OO(d \log \log n)$ time.
    
\end{lem}

\begin{proof}
    We assume without loss of generality that both $d, t$ are powers of $2$.
    Computing the binary matrix representation of $h(x)$ can be done in $\OO(d)$ time since $t = \OO(\log n)$. Given this, we compute $h(x)^\top$ as follows.

\ \\
{\em Case 1: $d \geq t$:}  We partition the matrix into $t$-by-$t$ square submatrices, denoted by 
\[ \mathsf{Matrix}(h(x)) := \underbrace{\begin{bmatrix}
        \begin{array}{c}
           M_1 \\
            \hline
            M_2\\
            \hline
            \vdots\\
            \hline
            M_{d/t}
        \end{array}
        \end{bmatrix}}_{\displaystyle t}\left.\vphantom{\begin{bmatrix}
        \begin{array}{c}
           M_1 \\
            \hline
            M_2\\
            \hline
            \vdots\\
            \hline
            M_{d/t}
        \end{array}
        \end{bmatrix}}\right\}d
        \] 
        For each $i \in [d/t]$, we use a recursive subroutine $\texttt{Transpose}(M_i, t)$ to compute $M_i^\top$. See Figure \ref{Figure:transposePic} for a pictorial illustration of the $\texttt{Transpose}$ algorithm.

        \begin{figure}
            \centering
            \includegraphics[width=0.95\linewidth]{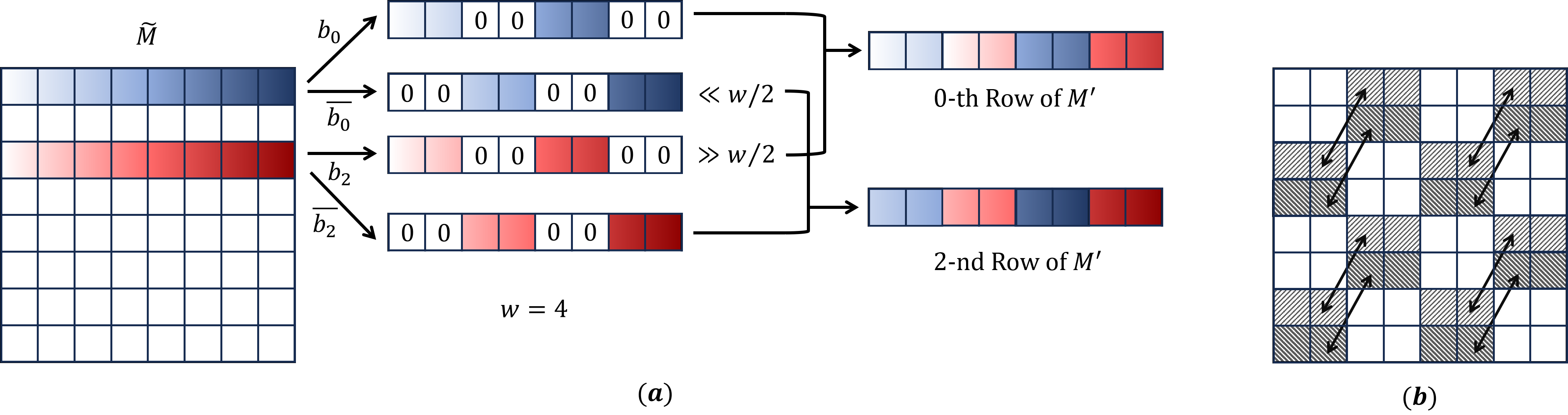}
            \caption{A pictorial example for $8$-by-$8$ square matrix and $w=4$. The left figure (a) shows how the $\texttt{Transpose}$ algorithm handles the rows of $\tilde{M}$ in lines \ref{Line:string} and \ref{Line:shift}. The right figure (b) illustrates the transpose outcome.}
            \label{Figure:transposePic}
        \end{figure}
\begin{figure}[ht!]
        \begin{customizedFrame}{\texttt{Transpose}}
		
		\begin{flushleft}
			\noindent {\bf Input:} An $I$-by-$J$ bit matrix $M$ where $I \leq J$ are powers of $2$. An integer $w$ that is a power of $2$ and $2 \leq w \leq I$.
			
			\noindent {\bf Output:} An $I$-by-$J$ matrix $M'$ such that if it is partitioned into $w$-by-$w$ square submatrices, then each submatrix is the transpose of the corresponding submatrix of $M$ at the same coordinates.

            \begin{enumerate}
            \item Let \[ \tilde{M} = \begin{cases}
                M & \text{if $w$=2}\\
                 \texttt{Transpose}(M, w/2) & \text{otherwise}
            \end{cases}
            \] be zero-indexed and $\tilde{M}[i, j]$ is its $(i, j)$-th entry.
            \item 
            For each integer $i$ such that $0 \leq i < I$:
            \begin{enumerate}
                \item\label{Line:string} Compute a $J$-bit binary string $b_i$ such that for $j: 0 \leq j < J$, its $j$-th bit $b_i[j] = \begin{cases}
                    \tilde{M}[i, j] &\text{if $(j \mod w) < w/2$} \\
                    0 &\text{otherwise}
                \end{cases}$.
                
                Also, compute a string  $\overline{b_i}[j] = \begin{cases}
                    0 &\text{if $(j \mod w) < w/2$} \\
                    \tilde{M}[i, j] &\text{otherwise}
                \end{cases}$.
            \end{enumerate}
            
            \item 
            Define an $I$-by-$J$ matrix $M'$, such that for each integer $0 \leq i < I$:
            \begin{enumerate}
                \item\label{Line:shift}   
                Let the $i$-th row of $M'$ be $\begin{cases}
                    b_i + b_{i+(w/2)} \gg (w/2) & \text{if $(i \mod w) < w/2$} \\
                    \overline{b_i} + \overline{b_{i-(w/2)}} \ll (w/2) & \text{if $(i \mod w) \geq w/2$}
                \end{cases}$,
                
                where $\gg (w/2)$ (resp. $\ll$) denote the operation of shifting a string to the right (resp. left) by $w/2$ bits.
                
            \end{enumerate}

            \item Output $M'$.
            \end{enumerate}
		\end{flushleft}
	\end{customizedFrame}
\end{figure}
    The correctness of the $\texttt{Transpose}$ algorithm can be shown by induction on (the base-$2$ logarithm of) $w$. When $I = J = t = \OO(\log n)$, Line \ref{Line:string} and \ref{Line:shift} can be done using a constant number of operations on words. Thus we get the following runtime.

    \begin{clm}
        Assuming $t = \OO(\log n)$ the procedure $\texttt{Transpose}(M_i, t)$ runs in $\OO(t \cdot \log t)$ time.
    \end{clm}

     We execute the $\texttt{Transpose}$ algorithm for all $i$, which takes $\OO((d/t) \cdot t\log t) = \OO(d \log\log n)$.  Then we can write down $h(x)^\top$ by concatenating rows of $M_i^\top$'s, which takes $\OO(t\cdot(d/t))$ time.

\ \\
{\em Case 2: $t > d$:} We again partition the matrix into $t$-by-$t$ square submatrices. In this case, we obtain $\mathsf{Matrix}(h(x)) = M := \underbrace{\begin{bmatrix}
        \begin{array}{c|c|c|c}
           M_1 & M_2 & \hdots & M_{t/d}
        \end{array}
        \end{bmatrix}}_{\displaystyle t}\left.\vphantom{\begin{bmatrix}
        \begin{array}{c|c|c|c}
           M_1 & M_2 & \hdots & M_{t/d}
        \end{array}
        \end{bmatrix}}\right\}d$. 
        
        \begin{clm}
        Given $t = \OO(\log n)$, $\texttt{Transpose}(M, d)$ runs in $\OO(d\log d) \leq \OO(d\log\log n)$ time.
    \end{clm}
    
    We execute $\texttt{Transpose}(M, d)$ and obtain $M' = \begin{bmatrix}
        \begin{array}{c|c|c|c}
           M_1^\top & M_2^\top & \hdots & M_{t/d}^\top
        \end{array}
        \end{bmatrix}$. In principle, to obtain $h(x)^\top$, we just concatenate $d\cdot (t/d)$ rows of all $M_i^\top$. However, when $t \gg d$, this takes longer than $\OO(d \log\log n)$ time. We instead use another recursive subroutine $\texttt{Concatenate}(M', d)$. An example of the  $\texttt{Concatenate}$ algorithm is given in Figure \ref{Figure:concatenatePic}.

        \begin{figure}
            \centering
            \includegraphics[width=0.9\linewidth]{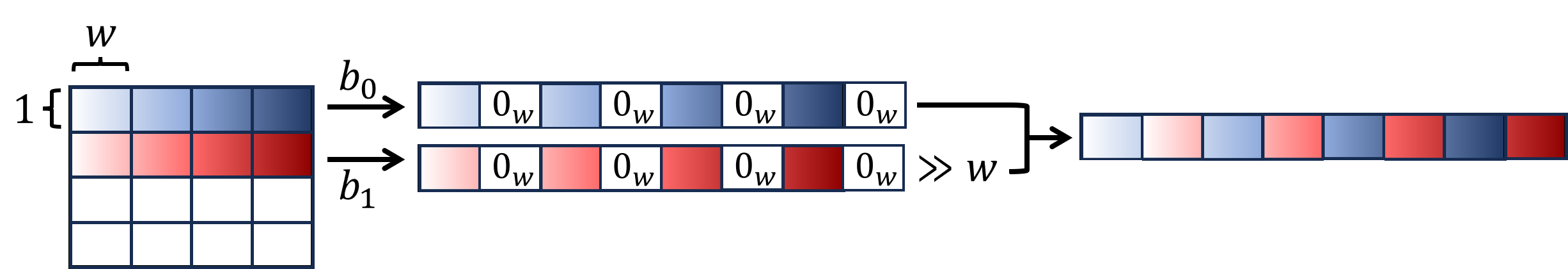}
            \caption{A pictorial example that shows the behavior of Line \ref{Line:spacing} and \ref{Line:merge} of the  $\texttt{Concatenate}$ algorithm on a $4$-by-$4w$ matrix.}
            \label{Figure:concatenatePic}
        \end{figure}

\begin{figure}[ht!]
        \begin{customizedFrame}{\texttt{Concatenate}}
		
		\begin{flushleft}
			\noindent {\bf Input:} An $I$-by-$J$ bit matrix $M$ where $I \leq J$ are powers of $2$. An integer $w$ that is a power of $2$ and $I \leq w \leq J$.
			
			\noindent {\bf Output:} An $IJ$-bit string $B$ such that if it is partitioned into $w$-bit blocks, then the $u$-th block (zero-indexed from left to right) are bits on the $(u \mod I)$-th row of $M$ from column $w\cdot \lfloor u/I\rfloor$ to $w\cdot \lfloor u/I \rfloor+w$.

            \begin{enumerate}
            \item If $w = I$ then output $B = M$.

            \item For each integer $i$ such that $0 \leq i < I$:
            \begin{enumerate}
                \item\label{Line:spacing} Partition the $i$-th row of $M$ into $w$-bit blocks, denoted as $\overbrace{[\underbrace{b_{i, 1}}_{w} \mid b_{i, 2} \mid \hdots \mid b_{i, J/w}]}^{J}$. Compute a $2J$-bit string $b_i = [b_{i, 1} \mid 0_w \mid b_{i, 2} \mid 0_w \mid \hdots \mid b_{i, J/w} \mid 0_w]$, where $0_w$ is a $w$-bit all-zero string.
            \end{enumerate}
            
            \item Define an $I/2$-by-$2J$ matrix $M'$, such that for each integer $0 \leq i < I/2$:
            \begin{enumerate}
                \item\label{Line:merge} 
                Let the $i$-th row of $M'$ be $b_{2i} + b_{2i+1} \gg w$, where $\gg w$ is the operation of shifting a string to the right by $w$ bits.
            \end{enumerate}
            
            \item Output $\texttt{Concatenate}(M', 2w)$.
            \end{enumerate}
		\end{flushleft}
	\end{customizedFrame}
    \end{figure}

The correctness of the $\texttt{Concatenate}$ algorithm can again be observed by inducting on the logarithm of $w$. Line \ref{Line:spacing} and \ref{Line:merge} can be done using $\OO((J/w) \cdot \lceil w/\log n\rceil)$ operations on words, and both lines are repeated for $I$ times in each recursive call. Therefore, the total runtime is 
\begin{align*}
\OO(\sum_{s = 1}^{\log (d)} 2^s \cdot \frac{2^{\log (d) - s}t}{2^{\log (d) - s}d} \cdot \lceil\frac{2^{\log (d) - s}d}{\log n} \rceil) &= \OO(\sum_{s = 1}^{\log (d)} 2^s  \cdot \max{(\frac{t}{d}, \frac{t}{d} \cdot \frac{2^{\log (d) - s}d}{\log n})}) \\
&= \OO(\log d\cdot \max(t, \frac{td}{\log n})) \\
&= \OO(d\log\log n)
\end{align*}
\end{proof}

\begin{theorem}
    The $\QuadTree$ algorithm runs in $\OO(nd \log\log n)$ time.
\end{theorem}

\begin{proof}
Computing $h(x)$ for all $x$ takes $\OO(nd)$ time. Then computing $h(x)^\top$ takes $\OO(nd \log \log n)$ time. After that, sorting $\OO(n)$-many $\OO(d\log n)$-bit strings can be done in $\OO(nd)$ time using radix sort. Finally, to find $p \in P$ with the longest common prefix for every $q \in Q$, we go through the sorted list and link each $q \in Q$ with adjacent $p \in P$, which takes $\OO(n)$ total time. The above time bounds also hold for $h'(x)$'s, resulting in $\OO(nd \log\log n)$ time in total.
    
\end{proof}
\section{Tournament}
\label{s:tour}
In this section, we compute the $2$-approximation of the nearest neighbor distances for logarithmically many queries. We do so using a depth-$2$ tournament: we first partition input points into random groups, project them to a lower dimensional space, and collect the nearest neighbor in the projected space in every group as a set $\tilde{S}$. Then the final output of the tournament is the nearest neighbor among points in $\tilde{S}$ in the original space. Intuitively, because each random group is small,  the true nearest neighbor could only lose to another near neighbor in the first step. Then in the second step, $\tilde{S}$ should contain at least one near neighbor.

\ \\
{\bf Notation:} We use the same notation ${{D}} := \min{(d, 3\log n)}$ as in the previous section. For any finite subset $T \subset \R$, let $\med T \in \R$ denote the median of $T$.

When working under the $\ell_1$ norm, we use Cauchy random variables to project points. We first recall a standard bound on the median of projections, which will be useful for our analysis. (The following lemma essentially follows from Claim 2 and Lemma 2 in~\cite{indyk2006stable}; we reprove it in the appendix for completeness.)

\begin{lem}\label{Lemma:comparison}
    Let $x, y \in \R^d$ and $0 < c <1/2$. Sample $r$ random vectors $v_1, v_2, \cdots, v_r \sim (\mathsf{Cauchy}(0, 1))^d$. With probability at least $1 -2e^{-rc^2/50}$, $\med\{|v_i \cdot (x -y) | : i \in [r]\} \in (1\pm c) \lVert x - y\rVert_1$.
\end{lem}

In Figure \ref{Figure:Tournament}, we describe how to construct a data structure to find $2$-approximate nearest neighbors. The construction borrows ideas from the second algorithm of \cite{K97}, but using a tournament of depth $2$ instead of $\OO(\log n)$.

\begin{figure}[ht!]
	\begin{customizedFrame}{\Tournament}
		
		\begin{flushleft}
			\noindent {\bf Input:} A set of $t$ queries $\{q_i\}_{i \in [t]}$ and a set of $n$ points $P$, which are both subsets of a metric space $(\R^d, \|\cdot\|_1)$.
			
			\noindent {\bf Output:} A set of $t$ values $\{\calD_i\}_{i \in [t]}$, such that every $\calD_i \in \R$ satisfies $\calD_i \geq \opt^P_{q_i}$.

            \paragraph{Building the Data Structure.}

            \begin{enumerate}
				\item Let $r \geq 800(2\log t + \log\log n)$.
                
                \item For each $j \in [r]$, draw  $v_j \sim (\mathsf{Cauchy}(0, 1))^d$, compute $v_j \cdot p$ for all points $p \in P$, and store all $v_j$ and $v_j \cdot p$.

                \item\label{Line:partition} Randomly partition $P$ into $n/ \log n $ subsets $P_1, P_2, \cdots, P_{n/ \log n}$, each of size $\log n$.
            \end{enumerate}

            \paragraph{Processing the  Queries.} For each query $q := q_i$ for $i \in [t]$: 


            \begin{enumerate}
            \item Compute $v_j \cdot q$ for all $j \in [r]$. 


\item\label{Line:add} Let $\tilde{S}$ be an empty set. For each $k \in [n/\log n]$:

\begin{itemize}
    \item Compute $\med_{p} := \med\{|v_j \cdot (q -p) | : j \in [r]\}$ for every $p \in P_k$.

    \item Find $p := \argmin_{p \in P_k}\{\med_{p}\}$ and add it into $\tilde{S}$.
\end{itemize}

\item \label{Line:output} 
Output $\calD_i := \min_{p \in \tilde{S}} \lVert q- p \rVert_1$ by computing and comparing all {\it exact} distances $\lVert q - p\rVert_1$ for $p \in \tilde{S}$.


        \end{enumerate}
			
		\end{flushleft}
	\end{customizedFrame}
	\caption{The $\Tournament$ Algorithm.}\label{Figure:Tournament}
\end{figure}

\ \\
{\bf Correctness:}

We fix a query $q := q_i$. 




\begin{lem}\label{Lemma:small}
   With probability at least $1- \frac{1}{10t}$, $\calD_i \leq 2\opt_q$. 
\end{lem}

We let $S$ denote the set of all $2$-approximate nearest neighbors to $q$, i.e., $S := \{p \in P : \lVert q-p \rVert_1 \leq 2 \opt_{q} \}$, and let $p^*\in P$ denote a nearest neighbor of $q$, i.e. $\lVert q-p^*\rVert_1 = \opt_q$. 
To prove Lemma \ref{Lemma:small}, we first make the following observation: 

\begin{lem}\label{Lemma:unique}
   Let $P'$ be an  arbitrary subset of $P \setminus S$. The probability that there exists $p \in P'$ such that $\med_p \leq \med_{p^*}$, where $\med_{p} := \med\{|v_j \cdot (q -p) | : j \in [r]\}$ and $\med_{p^*} := \med\{|v_j \cdot (q -p^*) | : j \in [r]\}$, is at most  $\frac{2(|P'|+1)}{t^2\log n}$. 
\end{lem}

\begin{proof}
    From $P' \subseteq (P \setminus S)$ we know that $\lVert q - p\rVert_1  > 2 \lVert q - p^*\rVert_1 $ for any $p \in P'$. Therefore, if $\med_p \leq \med_{p^*}$ then either $\med_p \neq (1\pm \frac{1}{4}) \lVert q - p\rVert_1$ or $\med_{p^*} \neq (1\pm \frac{1}{4}) \lVert q - p^*\rVert_1$. Applying Lemma \ref{Lemma:comparison} with $c = \frac{1}{4}$ and $r \geq 800(2\log t + \log\log n)$ and a union bound, we get that 
    \begin{align*}
        \pr{\exists p \in P': \med_p \leq \med_{p^*}} &\leq \pr{\exists p \in P':  \med_p \neq (1\pm \frac{1}{4}) \lVert q - p\rVert_1} + \\ & \hspace{1.3em} \pr{\med_{p^*} \neq (1\pm \frac{1}{4}) \lVert q - p^*\rVert_1} \\
        &\leq (|P'| +1) \cdot 2e^{-(2\log t + \log\log n)}\\
        &= \frac{2(|P'|+1)}{t^2\log n}.
    \end{align*}
\end{proof}

In Line $\ref{Line:partition}$ of the data structure building procedure, the point $p^*$ is assigned to one of the subsets $P^* \in \{P_1, P_2, $ $\cdots, P_{n/\log n}\}$. Focusing on this subset $P^*$, we can show that with high probability, either $p^*$ is added into $\tilde{S}$, or $p^*$ loses to another 2-approximate nearest neighbor. In both cases, the data structure is guaranteed to output an 2-approximation.

\begin{proof}[Proof (of Lemma \ref{Lemma:small}).]
    $P^* \setminus S$ contains at most $|P^*| = \log n$ points. Applying Lemma \ref{Lemma:unique}, we get $\med_{p^*} \geq \med_{p}$ simultaneously for all $p \in P^* \setminus S$ with probability at least $1- \frac{2\log n}{t^2 \log n} \geq 1- \frac{1}{10t}$ (as long as $t \geq 20$).  Conditioned on this, $\argmin_{p \in P^*} \{\med_p\}$ must be either $p^*$ or some other element of $S$. Thus on line $\ref{Line:output}$ of the algorithm, $\tilde{S}$ contains at least one element of $S$, so the final output $\min_{p \in \tilde{S}} \| q-p\|_1 \leq 2\opt_q$.
    
\end{proof}

Applying a union bound on Lemma \ref{Lemma:small}, we get the correctness guarantee:

\begin{theorem}
    Given $t$ queries $\{q_i\}_{i \in [t]}$, with probability at least $9/10$, the $\Tournament$ algorithm outputs $2$-approximate nearest neighbors simulataneously for all $t$ queries.
\end{theorem}

Finally, we state the runtime guarantee as follows:

\begin{theorem}\label{Theorem:tournamentTime}
    The $\Tournament$ algorithm runs in $\OO(n (d+ t) (\log t + \log \log n) + dt^2 \log t \log n) $ time.
\end{theorem}

\begin{proof}
    For preprocessing, the algorithm projects all points in $P$ using $r$ projections, which takes $\OO(n \cdot d \cdot r)$ time.
    To process a query $q$, we first take $\OO(dr)$ time to project $q$. We then count the number of comparisons we make to find the minimums of medians, which is $\OO((n/\log n)\log n \cdot r)$ using a linear-time median selection algorithm \cite{BFP73}. Each comparison can be done in $\OO(1)$ time given that $v_j \cdot p$ and $v_j \cdot q$ for all $j \in [r]$ and $p \in P$ are stored. Finally, we use $\OO(d \log n)$ time to do a linear scan over $\tilde{S}$.
    
    We plug in $r = \OO(\log t + \log \log n)$.
    For $t$ queries, the total runtime is $\OO(n (d+ t) (\log t + \log \log n) + dt \log n) $.
\end{proof}

For our purpose of estimating the Chamfer distance, we will apply the $\Tournament$ algorithm with a number of queries $t = \Theta({{D}}/\eps^2)$ for ${{D}} = \min{(d, 3\log n)}$ and some $\eps > 0$ satisfying $\eps^{-2} = \OO(\frac{n}{\log n})$. Under this setting, the runtime is dominated by the first additive term of Theorem \ref{Theorem:tournamentTime}, which is at most $\OO(n d(\log\log n+\log\frac{1}{\eps})/\eps^2)$.



\section{Rejection Sampling}
\ \\
{\bf Notation:} All occurrences of $\opt$ in this section are with respect to the set $B$. Let $\eps > 0$ be our target approximation factor. We call the distribution $\calP$ an $f$-Chamfer distribution for some $f = f(n, d, \eps)$, if it is supported on $A$ and for every $a \in A$,
\[f\frac{\opt_a}{\mathsf{CH}(A, B)} \leq \calP(a), \text{where we denote } \calP(a) := \ppr{x \sim \calP}{x = a}.\] 

We first show a general bound for estimating the Chamfer distance using samples from a Chamfer distribution. This follows from a standard analysis of importance sampling. 
\begin{lem}\label{Lemma:general}
    Let $X := \{x_i\}_{i \in [t]}$ be a set of $t$ samples drawn from a $f$-chamfer distribution $\calP$. Fix $h = h(n, d, \eps) \geq 1$. Given an arbitrarily $\tilde{\opt}_{x_i}$ for every $x_i$ that satisfies $\opt_{x_i} \leq \tilde{\opt}_{x_i} \leq h\cdot \opt_{x_i}$, then for any $0 < \kappa < 1$,
    \[\bpr{\tilde{\mathsf{CH}}(A, B) \leq (1-\kappa) \mathsf{CH}(A, B) } +  \bpr{\tilde{\mathsf{CH}}(A, B) \geq (1+\kappa) \cdot h \cdot \mathsf{CH}(A, B)} \leq \frac{\frac{h^2}{f}-1}{t \cdot \kappa^2}, \] where
    $\tilde{\mathsf{CH}}(A, B) := \frac{\sum_{i \in [t]}\tilde{\opt}_{x_i} / \calP(x_i) }{t }.$
\end{lem}

\begin{proof}
For the purpose of analysis, assume that we additionally have arbitrary $\tilde{\opt}_{a}$ for $a \in (A \setminus X)$ that also satisfies $\opt_{a} \leq \tilde{\opt}_{a} \leq h\cdot \opt_{a}$. 
    By linearity,
    \[\e{\tilde{\mathsf{CH}}(A, B)} = \frac{\sum_{i \in [t]}\e{ \tilde{\opt}_{x_i} / \calP(x_i) }}{t}
        =  \sum_{a \in A} \calP(a) \cdot \frac{\tilde{\opt}_{a}}{\calP(a) }
        \in [\mathsf{CH}(A, B), h\cdot \mathsf{CH}(A, B)].\]
    We also bound the variance
    \begin{align*}
        \var{\tilde{\mathsf{CH}}(A, B)} &\leq \frac{\e{\tilde{\opt}_{x_1}^2 / \calP(x_1)^2}}{t} - \mathsf{CH}(A, B)^2\\
        &\leq \frac{1}{t} (\sum_{a \in A}  \frac{\tilde{\opt}_{a}^2}{\calP(a)} - \mathsf{CH}(A, B)^2) \\
        &\leq \frac{1}{t}(\frac{h}{f}\mathsf{CH}(A, B)\cdot \sum_{a \in A}  \tilde{\opt}_{a} - \mathsf{CH}(A, B)^2) \\
        &\leq \frac{1}{t}\cdot \mathsf{CH}(A, B)^2 \cdot (\frac{h^2}{f}- 1) 
    \end{align*}
    where the third inequality follows from $\frac{1}{\calP(a)} \leq \frac{\mathsf{CH}(A, B)}{f\cdot  \opt_a}$ and $\tilde{\opt}_a \leq  h \cdot \opt_a$. Finally, by Chebyshev's Inequality, we have
    \[\bpr{\Bigl| \tilde{\mathsf{CH}}(A, B) - \e{\tilde{\mathsf{CH}}(A, B)} \Bigr| \geq \kappa \cdot \mathsf{CH}(A, B)} \leq \frac{1}{t} \cdot \frac{\frac{h^2}{f}-1}{\kappa^2 }.\]
\end{proof}

In this section, we aim to construct a set of samples $S = \{s_j\}_{j \in [s]}$ for some large enough $s$, such that
each $s_j$ is drawn from a fixed $\OO(1)$-Chamfer distribution. Once we have $S$, we can compute a weighted sum of the nearest neighbor distances for $s_j \in S$, and invoke Lemma \ref{Lemma:general} to show that it is likely an $(1+\eps)$-estimation of $\mathsf{CH}(A, B)$.

We will construct such $S$ via a two-step sampling procedure: in the first step, we sample $\Theta({{D}}/\eps^2)$ points from $A$ using a distribution defined by the estimations from the $\QuadTree$ algorithm. In the second step, we subsample these $\Theta({{D}}/\eps^2)$ points, using an acceptance probability defined by the estimations from the $\Tournament$ algorithm. We describe our \textbf{Chamfer-Estimate} algorithm in Figure \ref{Figure:Chamfer}.

\begin{figure}[ht!]
	\begin{customizedFrame}{Chamfer-Estimate}
		
		\begin{flushleft}
			\noindent {\bf Input:} Two subsets $A,B$ of a metric space $(\R^d, \|\cdot\|_1)$ of size $n$, a parameter $\eps > 0$, and a parameter $q \in \N$.
			
			\noindent {\bf Output:} An estimated value  $\tilde{\mathsf{CH}}(A, B) \in \R$.
			
			\begin{enumerate}
				\item\label{Line:QuadTree} Execute the algorithm $\QuadTree(A, B)$, and let the output be a set of values $\{ \calD_a \}_{a \in A}$ which always satisfy $\calD_a \geq \opt_a$. Let $\calD := \sum_{a \in A} \calD_a $. 
                
				\item\label{Line:distribution} Construct a probability distribution $\calP$ supported on $A$ such that for every $a \in A$, $\calP(a) = \frac{\calD_a}{\calD}$. For $i \in [q]$, sample $x_i \sim \calP$.
                
				\item\label{Line:Tournament} Execute the algorithm $\Tournament(\{x_i\}_{i \in [q]}, B)$, and let the output be a set of values $\{ \calD'_{x_i} \}_{i \in [q]}$ which always satisfy $\calD'_{x_i} \geq \opt_{x_i}$. Let $\calD' := \sum_{i \in [q]} \frac{\calD'_{x_i}}{\calP(x_i)} / q$ and denote $\calP'(a) := \frac{\calD'_{a}}{\calD'}$ (which is well-defined only if $a = x_i$ for some $i \in [q]$).
                
				\item\label{Line:sample} Define
                \[M := \max_{i \in [q]} \frac{\calP'(x_i)}{\calP(x_i)}.\]
                For each $i \in [q]$, mark $x_i$ as \textsc{accepted} with probability $\frac{\calP'(x_i)}{M \cdot \calP(x_i)}$. 
                
                If the number of \textsc{accepted} $x_i$ is less than $s=10/\eps^2$ then output \textbf{Fail} and exit the algorithm. Otherwise, collect the first $s$ \textsc{accepted} $x_i$ as a set $S := \{s_j\}_{j \in [s]}$.

                \item\label{Line:exact} Compute $\opt_{s_j}$ for each $j$. Output \[\tilde{\mathsf{CH}}(A, B) := \sum_{j \in [s]} \frac{\opt_{s_j}}{\calP'(s_j)} / s.\]
			\end{enumerate}
		\end{flushleft}
	\end{customizedFrame}
	\caption{The \textbf{Chamfer-Estimate} Algorithm.}\label{Figure:Chamfer}
\end{figure}

The \textbf{Chamfer-Estimate} algorithm applies the $\QuadTree$ algorithm and the $\Tournament$ algorithm as subroutines. If they are executed successfully, their outputs should satisfy the following conditions:

\begin{cond}\label{Condition:QuadTree}
    We say the $\QuadTree$ algorithm succeeds if for every $a \in A$, $\e{\calD_a} \leq 5{{D}} \cdot \opt_{a}$.
\end{cond}

\begin{cond}\label{Condition:Tournament}
    We say the $\Tournament$ algorithm succeeds if for every $x_i$ for $i \in [q]$, $\calD'_{x_i} \leq 2\opt_{x_i}$. 
\end{cond}

That is,  as described in the introduction, we need $\QuadTree$ to provide $\OO(\log n)$-approximation (to ensure that the sample size $q$ can be at most logarithmic in $n$), and that $\Tournament$ provide $\OO(1)$-approximation (to ensure that the final estimator using $s$ samples has variance bounded by a constant). 

We state some facts about the \textbf{Chamfer-Estimate} algorithm, which will be useful for our analysis. 


\begin{clm}[Line \ref{Line:distribution}]\label{Claim:P}
    Under Condition \ref{Condition:QuadTree}, with probability at least $9/10$, $\calP$ is a $(1/50{{D}})$-Chamfer Distribution.
\end{clm}

\begin{proof}
    With probability at least $9/10$, $\calD \leq 50{{D}}\cdot \mathsf{CH}(A, B)$ by Markov's Inequality. Upon this condition, for any $a \in A$,
    $\frac{\opt_{a}}{50{{D}}\cdot  \mathsf{CH}(A, B)} \leq \frac{\calD_a}{\calD}.$
\end{proof}

\begin{clm}[Line \ref{Line:Tournament}]\label{Claim:D'}       Let $q \geq 10^4 {{D}}$. Under Condition \ref{Condition:QuadTree} and \ref{Condition:Tournament}, with probability at least $4/5$, $\calD' \geq \mathsf{CH}(A, B)/2$. 
\end{clm}

\begin{proof}
    We apply the importance sampling analysis in Lemma \ref{Lemma:general}. We assume that Claim \ref{Claim:P} holds and $\opt_{x_i} \leq \calD'_{x_i} \leq 2\opt_{x_i}$, then
    \[\pr{\calD' \leq (1-\frac{1}{2})\mathsf{CH}(A, B)} \leq \frac{2^2\cdot 50{{D}} - 1}{q \cdot (\frac{1}{2})^2} <\frac{1}{10}.\]
\end{proof}

{\bf Analysis of $S$:} We now show that the set $S$ on Line \ref{Line:sample} collects enough samples (thus the algorithm does not fail) and is equivalent to sampling from a $\OO(1)$-Chamfer distribution $\calQ$. We note that the algorithm, in fact, only knows a $(1/50{{D}})$-Chamfer distribution $\calP$ and probabilities $\calP'(x_i)$ for $\{x_i\}_{i \in [q]}$, so it cannot explicitly sample from such $\calQ$. Nevertheless, by a standard analysis of rejection sampling, we show that $S$ ``simulates'' sampling from $\calQ$.

\begin{lem}\label{Lemma:size}
    Let $q \geq 10^4{{D}}/\eps^2$. Under Condition \ref{Condition:QuadTree} and \ref{Condition:Tournament}, with probability at least $3/5$, the number of \textsc{accepted} $x_i$ is at least $s$, so the algorithm does not fail.
\end{lem}

\begin{proof}
    We assume that Claim \ref{Claim:P} and \ref{Claim:D'} hold. Then for any $x_i$, $\frac{1}{\calP(x_i)} \leq \frac{50{{D}} \cdot \mathsf{CH}(A, B)}{\opt_{x_i}}$ and $\calP'(x_i) =\frac{\calD'_a}{\calD'} \leq \frac{2\opt_{x_i}}{\mathsf{CH}(A,B)/2}$. Thus $M \leq 200{{D}}$.
    The expectation is 
    \[ 
        \e{|\{\textsc{accepted } x_i\}|} = \sum_{i \in [q]} \frac{\calP'(x_i)}{M \calP({x_i})} 
        \geq \frac{1}{200{{D}}}  \sum_{i \in [q]} \frac{\calD'_{x_i}}{\calP(x_i)}\cdot\frac{1}{\calD'}
        = \frac{1}{200{{D}}} \cdot  q \calD'\cdot\frac{1}{\calD'}
        = \frac{q}{200 {{D}}}
    \]
    where the second to last equality is due to the definition of $\calD' := \sum_{i \in [q]}\frac{\calD'_{x_i}}{\calP(x_i)}/ q$. The final bound holds by Markov's Inequality and our setting of $q$.
\end{proof}


\begin{lem}\label{Lemma:equivalent}
    Each $s_j$ is independently and identically distributed, and under Condition \ref{Condition:Tournament}, $\pr{s_j = a} \geq \frac{\opt_a}{2\mathsf{CH}(A, B)}$ for any $a \in A$.
\end{lem}

\begin{proof}
    The independence and identicality follows directly from our sampling procedure. For the probability statement, we assume (without loss of generality) that during rejection sampling on Line \ref{Line:sample}, a sample $x_i$ is accepted and renamed as $s_j$. Then for any $a \in A$,
    \begin{align*}
        \pr{s_j = a} &= \pr{x_i = a \mid x_i\text{ accepted}} \\
        &= \frac{\calP(a) \cdot \pr{x_i\text{ accepted} \mid x_i = a}}{\pr{x_i\text{ accepted}}} \\
        &= \frac{\calP(a) \cdot \pr{x_i\text{ accepted} \mid x_i = a}}{\sum_{a_0 \in A} \calP(a_0) \cdot \pr{x_i\text{ accepted} \mid x_i = a_0}} \\
        &= \frac{\calP(a) \cdot \frac{\calP'(a)}{M\calP(a)}}{\sum_{a_0 \in A} \calP(a_0) \cdot \frac{\calP'(a_0)}{M\calP(a_0)}}
    \end{align*}
In the final equality, because we conditioned on $x_i = a$ (resp. $x_i = a_0$) on the LHS, we know that on the RHS, $\calP'(a) = \frac{\calD'(a)}{\calD'}$ is well-defined and satisfy $\opt_a \leq \calD'(a) \leq 2\opt_a$ (resp. $\opt_{a_0} \leq \calD'(a_0) \leq 2\opt_{a_0}$), given Condition \ref{Condition:Tournament}. Therefore, we have
\[\pr{s_j = a} = \frac{\calP'(a)}{\sum_{a_0 \in A}\calP'(a_0)} = \frac{\calD'(a)}{\sum_{a_0 \in A}\calD'(a_0)} \geq \frac{\opt_a}{2\mathsf{CH}(A, B)}.\]
\end{proof}

Lemma \ref{Lemma:size} and \ref{Lemma:equivalent} together say that $S$ can be viewed as a set of $s$ samples from a $\frac{1}{2}$-Chamfer Distribution, thus we can invoke another importance sampling analysis. In the final step of the algorithm, we compute the exact nearest neighbor distance for all $s_j$ and then compute a weighted sum over them. With high probability, this gives an $(1\pm\eps)$-estimation of $\mathsf{CH}(A, B)$.

\begin{theorem}
    Under Condition \ref{Condition:QuadTree} and \ref{Condition:Tournament}, {\normalfont \textbf{Chamfer-Estimate}($A, B, \eps, q \geq 10^4{{D}}/\eps^2$)} outputs $\tilde{\mathsf{CH}}(A, B)$ that satisfies $(1-\eps)\mathsf{CH}(A, B) \leq \tilde{\mathsf{CH}}(A, B) \leq (1+\eps)\mathsf{CH}(A, B)$ with probability at least $1/2$.
\end{theorem}

\begin{proof}
    In the success case of Lemma \ref{Lemma:size}, we can apply Lemma \ref{Lemma:general} with $f = 1/2$, $h = 1$, $t = s$, and $\kappa = \eps$. Then 
    \[\bpr{\Bigl| \tilde{\mathsf{CH}}(A, B) - \mathsf{CH}(A, B) \Bigr| \geq \eps \cdot \mathsf{CH}(A, B)} \leq \frac{1}{s \cdot \eps^2}. \]
\end{proof}

\begin{theorem}
    {\normalfont \textbf{Chamfer-Estimate}($A, B, \eps, q = 10^4{{D}}/\eps^2$)} runs in time $\OO(nd(\log\log n + \log \frac{1}{\eps})/\eps^2))$.
\end{theorem}

\begin{proof}
    This is dominated by the runtime of $\QuadTree$, $\Tournament$, and the time of computing $\opt_{s_j}$ on Line \ref{Line:exact}. $\QuadTree(A, B)$ runs in $\OO(nd \log \log n)$ time and $\Tournament$ $(\{x_i\}_{i \in [q]}, B)$ runs in $\OO(n d(\log\log n+\log\frac{1}{\eps})/\eps^2)$ time. Finally, the brute-force search for $\opt_{s_j}$ for $j \in [10/\eps^2]$ takes $\OO(nd/\eps^2)$ time.
\end{proof}

\nocite{langley00}

\section*{Acknowledgements}
Ying Feng was supported by an MIT Akamai Presidential Fellowship. Piotr Indyk was supported in part by the NSF TRIPODS program (award DMS-2022448).The authors would like to thank Anders Aamand for helpful comments that helped simplify the algorithm in Section~\ref{s:tour}.

\def\shortbib{0}
\bibliographystyle{alpha}
\bibliography{ref}

\newcommand{\etalchar}[1]{$^{#1}$}
\begin{thebibliography}{SMFW04}

\bibitem[AHNR95]{AHN95}
Arne Andersson, Torben Hagerup, Stefan Nilsson, and Rajeev Raman.
\newblock Sorting in linear time?
\newblock In {\em Proceedings of the Twenty-Seventh Annual ACM Symposium on Theory of Computing}, STOC '95, pages 427--436, New York, NY, USA, 1995. Association for Computing Machinery.

\bibitem[AM19]{atasu19a}
Kubilay Atasu and Thomas Mittelholzer.
\newblock Linear-complexity data-parallel earth mover’s distance approximations.
\newblock In Kamalika Chaudhuri and Ruslan Salakhutdinov, editors, {\em Proceedings of the 36th International Conference on Machine Learning}, volume~97 of {\em Proceedings of Machine Learning Research}, pages 364--373. PMLR, 09--15 Jun 2019.

\bibitem[AR15]{andoni2015optimal}
Alexandr Andoni and Ilya Razenshteyn.
\newblock Optimal data-dependent hashing for approximate near neighbors.
\newblock In {\em Proceedings of the forty-seventh annual ACM symposium on Theory of computing}, pages 793--801, 2015.

\bibitem[AS03]{athitsos2003estimating}
Vassilis Athitsos and Stan Sclaroff.
\newblock Estimating 3d hand pose from a cluttered image.
\newblock In {\em 2003 IEEE Computer Society Conference on Computer Vision and Pattern Recognition, 2003. Proceedings.}, volume~2, pages II--432. IEEE, 2003.

\bibitem[BFP{\etalchar{+}}73]{BFP73}
Manuel Blum, Robert~W. Floyd, Vaughan Pratt, Ronald~L. Rivest, and Robert~E. Tarjan.
\newblock Time bounds for selection.
\newblock {\em J. Comput. Syst. Sci.}, 7(4):448--461, August 1973.

\bibitem[BIJ{\etalchar{+}}23]{BIJ24}
Ainesh Bakshi, Piotr Indyk, Rajesh Jayaram, Sandeep Silwal, and Erik Waingarten.
\newblock A near-linear time algorithm for the chamfer distance, 2023.

\bibitem[Cha08]{C08}
Timothy~M. Chan.
\newblock Well-separated pair decomposition in linear time?
\newblock {\em Information Processing Letters}, 107(5):138--141, 2008.

\bibitem[FSG17]{fan2017point}
Haoqiang Fan, Hao Su, and Leonidas~J Guibas.
\newblock A point set generation network for 3d object reconstruction from a single image.
\newblock In {\em Proceedings of the IEEE conference on computer vision and pattern recognition}, pages 605--613, 2017.

\bibitem[HP11]{har2011geometric}
Sariel Har-Peled.
\newblock {\em Geometric approximation algorithms}.
\newblock Number 173. American Mathematical Soc., 2011.

\bibitem[Ind06]{indyk2006stable}
Piotr Indyk.
\newblock Stable distributions, pseudorandom generators, embeddings, and data stream computation.
\newblock {\em Journal of the ACM (JACM)}, 53(3):307--323, 2006.

\bibitem[JSQJ18]{jiang2018gal}
Li~Jiang, Shaoshuai Shi, Xiaojuan Qi, and Jiaya Jia.
\newblock Gal: Geometric adversarial loss for single-view 3d-object reconstruction.
\newblock In {\em Proceedings of the European conference on computer vision (ECCV)}, pages 802--816, 2018.

\bibitem[Kle97]{K97}
Jon~M. Kleinberg.
\newblock Two algorithms for nearest-neighbor search in high dimensions.
\newblock In {\em Proceedings of the Twenty-Ninth Annual ACM Symposium on Theory of Computing}, STOC '97, pages 599--608, New York, NY, USA, 1997. Association for Computing Machinery.

\bibitem[KSKW15]{kusner2015word}
Matt Kusner, Yu~Sun, Nicholas Kolkin, and Kilian Weinberger.
\newblock From word embeddings to document distances.
\newblock In {\em International conference on machine learning}, pages 957--966. PMLR, 2015.

\bibitem[KZ20]{khattab2020colbert}
Omar Khattab and Matei Zaharia.
\newblock Colbert: Efficient and effective passage search via contextualized late interaction over bert.
\newblock In {\em Proceedings of the 43rd International ACM SIGIR conference on research and development in Information Retrieval}, pages 39--48, 2020.

\bibitem[pda23]{pdal}
Pdal: Chamfer.
\newblock \url{https://pdal.io/en/2.4.3/apps/chamfer.html}, 2023.
\newblock Accessed: 2023-05-12.

\bibitem[pyt23]{pytorch3d}
Pytorch3d: Loss functions.
\newblock \url{https://pytorch3d.readthedocs.io/en/latest/modules/loss.html}, 2023.
\newblock Accessed: 2023-05-12.

\bibitem[SMFW04]{sudderth2004visual}
Erik~B Sudderth, Michael~I Mandel, William~T Freeman, and Alan~S Willsky.
\newblock Visual hand tracking using nonparametric belief propagation.
\newblock In {\em 2004 Conference on Computer Vision and Pattern Recognition Workshop}, pages 189--189. IEEE, 2004.

\bibitem[ten23]{tensorflow}
Tensorflow graphics: Chamfer distance.
\newblock \url{https://www.tensorflow.org/graphics/api_docs/python/tfg/nn/loss/chamfer_distance/evaluate}, 2023.
\newblock Accessed: 2023-05-12.

\bibitem[WCL{\etalchar{+}}19]{wan2019transductive}
Ziyu Wan, Dongdong Chen, Yan Li, Xingguang Yan, Junge Zhang, Yizhou Yu, and Jing Liao.
\newblock Transductive zero-shot learning with visual structure constraint.
\newblock {\em Advances in neural information processing systems}, 32, 2019.

\end{thebibliography}


\appendix

\section{Reducing the Bit Precision of Inputs.}

In our algorithm, we assumed that all points in input sets $A, B$ are integers in $\{1,2 , \cdots, $ $\poly(n)\}^d$. Here, we show that this is without loss of generality, as long as all coordinates of the original input are $w$-bit integers for arbitrary $w \geq \log n$ in a unit-cost RAM with a word length
of $w$ bits.

Section $A.3$ of \cite{BIJ24} gives an efficient reduction from real inputs to the case that \[1 \leq \min_{a \in A, b \in B} \lVert a - b\rVert_1  \leq \max_{a \in A, b \in B} \lVert a - b\rVert_1 \leq \poly(n),\] i.e., the input has a $\poly(n)$-bounded aspect ratio. Their reduction can be adapted to our case as follows:

\begin{clm}[Lemma $A.3$ of \cite{BIJ24}]
    Given an $\mathsf{est}$ such that $\mathsf{CH}(A, B) \leq  \mathsf{est} \leq \poly(n) \cdot \mathsf{CH}(A, B)$, if there exists an algorithm that computes an $(1+\eps)$-approximation to $\mathsf{CH}(A, B)$ in $\OO(nd(\log\log n+\log\frac{1}{\eps})/\eps^2))$ time under the assumption that $A, B$ contain points from 
    $\{1 \ldots \mbox{\poly}(n) \}^d$, then there exists an algorithm that computes an $(1+\eps)$-approximation to $\mathsf{CH}(A, B)$ for any 
 integer-coordinate $A, B$ in asymptotically same time.
\end{clm}

It remains to show how to obtain a $\poly(n)$-approximation.

\begin{lem}
    There exists an $\OO(nd + n\log\log n)$-time algorithm that computes $\mathsf{est}$ which satisfies $\mathsf{CH}(A, B) \leq  \mathsf{est} \leq \poly(n)\cdot \mathsf{CH}(A, B)$ with $1-\frac{1}{n}$ probability.
\end{lem}

\begin{proof}
    Similar to (the proof of Lemma $A.3$ in) \cite{BIJ24}, we sample a vector $v \sim \mathsf{Cauchy}(0, 1)$, which can be discretized to $\OO(\log n)$-bit precision following \cite{indyk2006stable}. We then compute the inner products $\{ v\cdot a \}_{a \in A}$ and $\{ v\cdot b \}_{b \in B}$. The distribution of $ v \cdot a - v \cdot b$ follows $\mathsf{Cauchy}(0, \lVert a- b\rVert_1)$ by the $1$-stability property of Cauchy's. So we have that for every $a \in A$ and $b \in B$, 
\[ \frac{\lVert a - b\rVert_1}{\poly(n)} \leq |v \cdot a - v \cdot b| \leq \lVert a - b\rVert_1 \cdot \poly(n), \]
with probability $1-1/\poly(n)$. Therefore, $\mathsf{est} := \mathsf{CH}(\{ v\cdot a \}_{a \in A}, \{ v\cdot b \}_{b \in B})$ is a $\poly(n)$-approximation to $\mathsf{CH}(A,B)$. We may assume by scaling that 
$\{ v\cdot a \}_{a \in A}, \{ v\cdot b \}_{b \in B}$ contain $w$-bit integers, which can be sorted in $\OO(n \log\log n)$ time \cite{AHN95}. Then to compute $\mathsf{est}$, we find all one-dimensional nearest neighbors by going
through the sorted list and link each $a'\in \{ v\cdot a \}_{a \in A}$ with adjacent $b'\in \{ v\cdot b \}_{b \in B}$, which takes $\OO(n)$ time. Thus the total runtime is $\OO(nd + n\log\log n)$ as claimed.
\end{proof}

\section{Proof of Lemma \ref{Lemma:comparison}}

\begin{proof}
    We use the fact that for $v \sim (\mathsf{Cauchy}(0, 1))^d$ and any $x \in \R^d$, $(v \cdot x) \sim \mathsf{Cauchy}(0, \lVert x \rVert_1)$. Also, for any $k > 0$, if a random variable $z \sim \mathsf{Cauchy}(0, 1)$ then $kz \sim \mathsf{Cauchy}(0, k)$. Therefore, for any $v_i : i \in [r]$, $\pr{|v_i \cdot (x - y)| > (1+c) \lVert x- y\rVert_1} = \pr{U > 1+c}$ where $U\sim \mathsf{HalfCauchy}(0, 1)$. The density of $U$ is $f_U(u) = \frac{2}{\pi}\cdot \frac{1}{1+u^2}$, thus $ \pr{U > 1} = 1/2$ and
    \begin{align*}
        \pr{U > 1+c} &= \frac{1}{2} - \int^{1+c}_1 f_U(u) du \\
        &\leq \frac{1}{2} - c \cdot f_U(3/2) && \text{for $0<c<1/2$}\\
        &< \frac{1}{2} - c/10
    \end{align*}

    Similarly, we can get $\pr{|v_i \cdot (x - y)| < (1-c) \lVert x- y\rVert_1} < \frac{1}{2} - c/10$. For $i \in [r]$, let $\mathcal{I}_i$ be an indicator variable that equals $1$ if $|v_i \cdot (x - y)| < (1-c) \lVert x- y\rVert_1$ and equals $0$ otherwise. By Hoeffding's bound,

    \[\pr{\sum_{i \in [r]} \mathcal{I}_i  \geq \frac{r}{2}} < e^{-2rc^2/100},\]

    which upper bounds the failure probability that the median is too small. We symmetrically bound the probability that the median is too large. Then \[\pr{\med\{|v_i \cdot (x -y) | : i \in [r]\} \in (1\pm c) \lVert x - y\rVert_1} \geq {1-2e^{-rc^2/50}}.\]
    
\end{proof}

\end{document}